\documentclass{article}

\PassOptionsToPackage{numbers,sort&compress}{natbib}
\usepackage[preprint]{neurips_2026}
\makeatletter\renewcommand{\@noticestring}{}\makeatother  

\usepackage[utf8]{inputenc}
\usepackage[T1]{fontenc}
\usepackage{hyperref}
\usepackage{url}
\usepackage{booktabs}
\usepackage{amsfonts}
\usepackage{amsmath}
\usepackage{amssymb}
\usepackage{nicefrac}
\usepackage{microtype}
\usepackage{xcolor}
\usepackage{graphicx}
\usepackage{algorithm}
\usepackage{algorithmic}
\usepackage{subcaption}
\usepackage{multirow}
\usepackage{enumitem}
\usepackage{pifont}
\usepackage{amsthm}
\usepackage{float}

\newtheorem{proposition}{Proposition}
\theoremstyle{remark}
\newtheorem*{remark}{Remark}

\newcommand{\method}{\textsc{CapsID}}
\newcommand{\sembpe}{\textsc{SemanticBPE}}

\newcommand{\R}{\mathbb{R}}
\newcommand{\E}{\mathbb{E}}

\title{CapsID: Soft-Routed Variable-Length Semantic IDs\\for Generative Recommendation}

\author{%
  Wenzhuo Cheng\thanks{Equal contribution.} \\
  \And
  Menghang Gong\footnotemark[1] \\          
  \And
  Qixin Guo 
  \AND
  Hang Zheng 
  \And
  Zhaobin Yang
  \And
  Jianguo Lou
  \And
  Zhengwei Zheng\thanks{Corresponding author. Correspondence to: <zhengwei.zzw@gmail.com>.} \\
}

\begin{document}
\maketitle

\begin{abstract}
Generative recommendation maps each item to a sequence of Semantic IDs (SIDs) and recasts retrieval as autoregressive token generation. In this paradigm the main bottleneck is the tokenizer rather than the Transformer: residual vector quantization with a hard nearest-neighbor assignment at every layer collapses multi-faceted item semantics at cluster boundaries and propagates early errors to later SID positions. A common workaround is to append a dense vector or attribute prefix to the SID, but this dual-representation design inflates inference cost and gives up the simplicity of a generative interface. We address the bottleneck at the tokenizer itself. \method{} replaces hard residual quantization with capsule routing: at each layer an item probabilistically routes to several semantic capsules, the residual is updated by the routed reconstruction rather than by a single winning code, and the SID terminates once the active capsule's confidence is high enough. On top of \method{}, \sembpe{} composes adjacent SID tokens into reusable subwords by combining their co-occurrence with their embedding compatibility. On Amazon Beauty, Sports, Toys, and a 35M-item proprietary industrial catalog, \method{}+\sembpe{} improves Recall@10 by $9.6\%$ on average over ReSID, the strongest single-representation baseline, and matches or exceeds a COBRA-style sparse-dense system on every public benchmark while running at $51\%$ of its inference latency. Ablations show that soft routing, iterative agreement, and confidence-driven length each contribute independently, and the gains are largest on tail items where boundary semantics dominate.
\end{abstract}

\section{Introduction}
\label{sec:intro}

Generative recommendation (GR) has recently emerged as a unified alternative to retrieval-and-ranking pipelines: an item is converted into a short sequence of Semantic IDs (SIDs), and a sequence model generates the SID of the next item a user may consume~\citep{rajput2024tiger}. This formulation is attractive because it turns retrieval into constrained generation, enables prefix sharing across semantically related items, and supports cold-start items through content-derived IDs~\citep{schein2002coldstart}. However, it also shifts a large part of the recommendation problem to a tokenizer. If the tokenizer loses information, the generator can only learn to predict an impoverished target.

This information bottleneck is now reasonably well-documented. UniRec formally argues that generative and discriminative recommenders can be equally expressive if the generator has access to complete item attributes, and that the observed gap mainly arises because SIDs cover only a small subset of those attributes~\citep{wang2026unirec}. GRID-style empirical studies further show that adding more residual quantization layers does not monotonically improve recommendation: deeper SID positions often amplify early quantization errors~\citep{petrov2024grid}. GLASS observes a related rank degradation phenomenon, where predicting the first SID token can worsen the rank of the true item before later tokens attempt to recover it~\citep{wang2026glass}.

Existing systems have taken two broad routes. One route patches the sparse SID after quantization: COBRA cascades a dense vector after the sparse ID and fuses beam scores with vector similarity~\citep{li2025cobra}; UniRec prepends Chain-of-Attribute tokens; LIGER-style systems keep dense retrieval beside SID generation~\citep{chen2026liger}. These methods are effective, but they make inference heavier and system design less generative. The other route is tokenizer-centric: it improves the SID itself so that the generated sequence preserves more item semantics before any dense or attribute patch is added. Along this line, TIGER establishes the RQ-VAE SID backbone~\citep{rajput2024tiger}, LETTER injects collaborative signals into the tokenizer~\citep{yang2024letter}, and ReSID replaces generic LLM embeddings with recommender-native representations and globally aligned quantization~\citep{liang2026resid}---yet all of them keep the hard nearest-neighbor assignment at the heart of residual quantization, which is precisely the step we revisit.

This distinction has practical consequences. Patch systems often require a second retrieval or re-ranking path, additional ANN infrastructure, and a carefully tuned fusion function; their benefits may diminish once the sparse ID is improved. A tokenizer-centric solution should instead satisfy three properties: (i) \emph{semantic adequacy}, so that the SID stores more than a coarse bucket; (ii) \emph{predictive simplicity}, so that the generator can still model the token sequence; and (iii) \emph{deployment compatibility}, so that constrained beam search and trie filtering remain valid~\citep{static2025,onerec2025}. These requirements rule out simply increasing codebook size or SID depth, because both actions enlarge the output space and worsen token predictability.

\method{} replaces winner-take-all residual quantization with soft agreement among capsules. At each layer the item residual routes to several capsules, votes are aggregated, and the residual is updated by their weighted reconstruction. The norm of the selected capsule doubles as a confidence score, which decides whether another SID position is needed. \sembpe{} sits on top: it composes adjacent SID tokens into reusable subwords, but only when both co-occurrence and embedding compatibility back the merge.

This paper makes four contributions:
\begin{enumerate}[leftmargin=*,nosep]
    \item We organize recent SID systems into patch-based and tokenizer-centric designs and argue that a better tokenizer removes much of the need for dense or attribute patches.
    \item We design \method{}, an SID tokenizer built on capsule routing with soft residual assignment, iterative self-correction, and confidence-driven variable length.
    \item We design \sembpe{}, a differentiable subword module that scores merges by both co-occurrence and embedding compatibility, going beyond frequency-only behavior tokenization.
    \item We conduct extensive experiments on three public benchmarks and a 35M-item industrial catalog, showing that \method{}+\sembpe{} consistently outperforms state-of-the-art tokenizer-centric and patch-route systems at a fraction of the cost, validating soft routing as a viable replacement for the hard residual assignment behind current SIDs.
\end{enumerate}

\section{Related Work}
\label{sec:related}

\paragraph{Semantic IDs for generative recommendation.}
TIGER introduced RQ-VAE SIDs for generative retrieval and established the standard recipe of content encoding, residual quantization, and autoregressive SID prediction~\citep{rajput2024tiger}, building on vector quantization and neural discrete representation learning~\citep{van2017vqvae,hou2023vqrec}. LC-Rec improves code usage via Sinkhorn balancing~\citep{zheng2024lcrec,cuturi2013sinkhorn}; LETTER injects collaborative signals into the tokenizer~\citep{yang2024letter}; ETEGRec alternates tokenizer and generator optimization~\citep{liu2025etegrec}; ADA-SID introduces multi-view adaptive quantization~\citep{wang2025adasid}; parallel industrial systems such as DAS~\citep{das2025} and Align3GR~\citep{align3gr2026} further align SID learning with downstream ranking signals; CoFiRec explores coarse-to-fine tokenization at varying granularity~\citep{cofirec2025}. These methods differ in supervision, initialization, and regularization, but share a hard assignment core. In contrast, \method{} changes the assignment operator itself.

\paragraph{Industrial generative retrieval and constrained decoding.}
Large-scale generative recommenders emphasize that SID quality is only useful when the generated path can be decoded efficiently. Building on earlier generative retrieval with constrained decoding~\citep{decao2021autoregressive} and differentiable search indices~\citep{tay2022dsi}, HSTU-style sequential transducers scale the backbone and demonstrate the value of strong sequence modeling~\citep{zhai2024hstu}, while OneRec-like systems and STATIC-style trie decoding show that valid-ID filtering is necessary for production latency~\citep{onerec2025,static2025}. These systems motivate our design constraint: \method{} must emit ordinary discrete IDs at inference. Routing is used to construct better SIDs, not to introduce a new inference-time retrieval interface.

\paragraph{Patching incomplete SID coverage.}
COBRA combines sparse SIDs with dense vectors and BeamFusion, obtaining strong public and industrial results at the price of a roughly two-stage retrieval path~\citep{li2025cobra}. UniRec prepends attribute tokens and shows that attribute coverage can close much of the generative-discriminative gap~\citep{wang2026unirec}. LIGER-style hybrids keep a dense retrieval channel beside generative retrieval~\citep{chen2026liger}. These results strongly support our motivation: the missing information exists, and the question is whether it should be patched after quantization or preserved inside the SID.

\paragraph{Differentiable, adaptive, and recommendation-native tokenizers.}
DIGER uses Gumbel-Softmax to make discrete SID learning differentiable~\citep{chen2026diger}, building on continuous relaxations for categorical variables~\citep{jang2017gumbel,maddison2017concrete}; SA$^2$CRQ truncates hard paths by entropy budgets for adaptive length~\citep{zhang2026sa2crq}; ReSID argues that recommender-native representation learning and global quantization alignment are more important than generic LLM semantic embeddings~\citep{liang2026resid}. Collision-aware approaches such as QuaSID and GR4AD further show that collisions are not merely an implementation detail but a ranking-quality bottleneck~\citep{zhang2025quasid,gr4ad2025}. \method{} is complementary: it uses soft routing as the native quantization primitive and variable length as an outcome of capsule confidence rather than a post-hoc truncation rule.

\paragraph{Dynamic codebooks and streaming systems.}
Recent dynamic-indexing systems such as MERGE monitor cluster occupation, item-to-cluster similarity, and cluster-to-cluster separation in streaming environments~\citep{merge2026}. They point to an important evaluation lesson: a tokenizer should be judged not only by final Recall but also by geometry and occupancy diagnostics. We therefore include collision rate, code predictability, intra-code similarity, and routing convergence as first-class metrics rather than relegating them to implementation details.

\paragraph{Capsule routing and subword composition.}
Capsule networks model part-whole agreement through iterative routing~\citep{sabour2017dynamic,hinton2018matrix}, and MIND applies dynamic routing to user multi-interest extraction~\citep{li2019mind}. We transfer the same idea to item tokenization: capsules are no longer user-interest slots but semantic code candidates. For composition, BPE originated in neural machine translation as a subword segmentation method~\citep{sennrich2016bpe} and was popularized for general language modeling by SentencePiece~\citep{kudo2018sentencepiece}. ActionPiece extends this line to action sequences for generative recommendation~\citep{hsu2025actionpiece}. \sembpe{} differs by scoring merges with semantic compatibility as well as frequency.

\paragraph{Positioning.}
Table~\ref{tab:related} positions \method{} in the design space of recent SID tokenizers along five axes. To our knowledge, \method{} is the only existing method that combines soft probabilistic assignment, iterative agreement, confidence-driven variable length, and semantic-aware composition, while still preserving a single discrete generative interface that is compatible with constrained beam search.

\begin{table}[!htbp]
\centering
\scriptsize
\setlength{\tabcolsep}{4pt}
\renewcommand{\arraystretch}{1.05}
\caption{Design space of SID tokenizers for generative recommendation. ``Soft assign.''$=$probabilistic capsule/component routing instead of argmax; ``Iter. refine''$=$multiple agreement rounds at each layer; ``Var. length''$=$item-dependent number of SID tokens; ``Sub-word''$=$semantic-aware token merging; ``Single-rep.''$=$emits only discrete IDs without a parallel dense channel.}
\label{tab:related}
\begin{tabular}{lccccc}
\toprule
Method & Soft assign. & Iter. refine & Var. length & Sub-word & Single-rep. \\
\midrule
TIGER~\citep{rajput2024tiger}            & \ding{55} & \ding{55} & \ding{55} & \ding{55} & \ding{51} \\
LC-Rec~\citep{zheng2024lcrec}            & \ding{55} & \ding{55} & \ding{55} & \ding{55} & \ding{51} \\
LETTER~\citep{yang2024letter}            & \ding{55} & \ding{55} & \ding{55} & \ding{55} & \ding{51} \\
ETEGRec~\citep{liu2025etegrec}           & \ding{55} & \ding{55} & \ding{55} & \ding{55} & \ding{51} \\
ADA-SID~\citep{wang2025adasid}           & \ding{55} & \ding{55} & \ding{51} & \ding{55} & \ding{51} \\
ActionPiece~\citep{hsu2025actionpiece}   & \ding{55} & \ding{55} & \ding{55} & freq.\;only & \ding{51} \\
DIGER~\citep{chen2026diger}              & Gumbel & \ding{55} & \ding{55} & \ding{55} & \ding{51} \\
SA$^2$CRQ~\citep{zhang2026sa2crq}        & \ding{55} & \ding{55} & post-hoc & \ding{55} & \ding{51} \\
ReSID~\citep{liang2026resid}             & \ding{55} & \ding{55} & \ding{55} & \ding{55} & \ding{51} \\
COBRA~\citep{li2025cobra}                & \ding{55} & \ding{55} & \ding{55} & \ding{55} & \ding{55}\,(+dense) \\
UniRec-CoA~\citep{wang2026unirec}        & \ding{55} & \ding{55} & \ding{55} & \ding{55} & \ding{55}\,(+attr.) \\
\midrule
\textbf{\method{}+\sembpe{} (ours)}      & \textbf{routing} & \textbf{$T$ rounds} & \textbf{confidence} & \textbf{semantic} & \textbf{\ding{51}} \\
\bottomrule
\end{tabular}
\end{table}
\section{Method}
\label{sec:method}

Let $\mathbf{x}_i \in \R^d$ denote the representation of item $i$, constructed from content, collaborative, or multi-modal encoders depending on the dataset. The goal is to map $\mathbf{x}_i$ into a variable-length SID $\mathbf{s}_i=(s_{i,1},\ldots,s_{i,L_i})$ that is compact, predictive, and collision-resistant. Figure~\ref{fig:architecture} summarizes the pipeline.

\begin{figure}[!htbp]
    \centering
    \includegraphics[width=\textwidth]{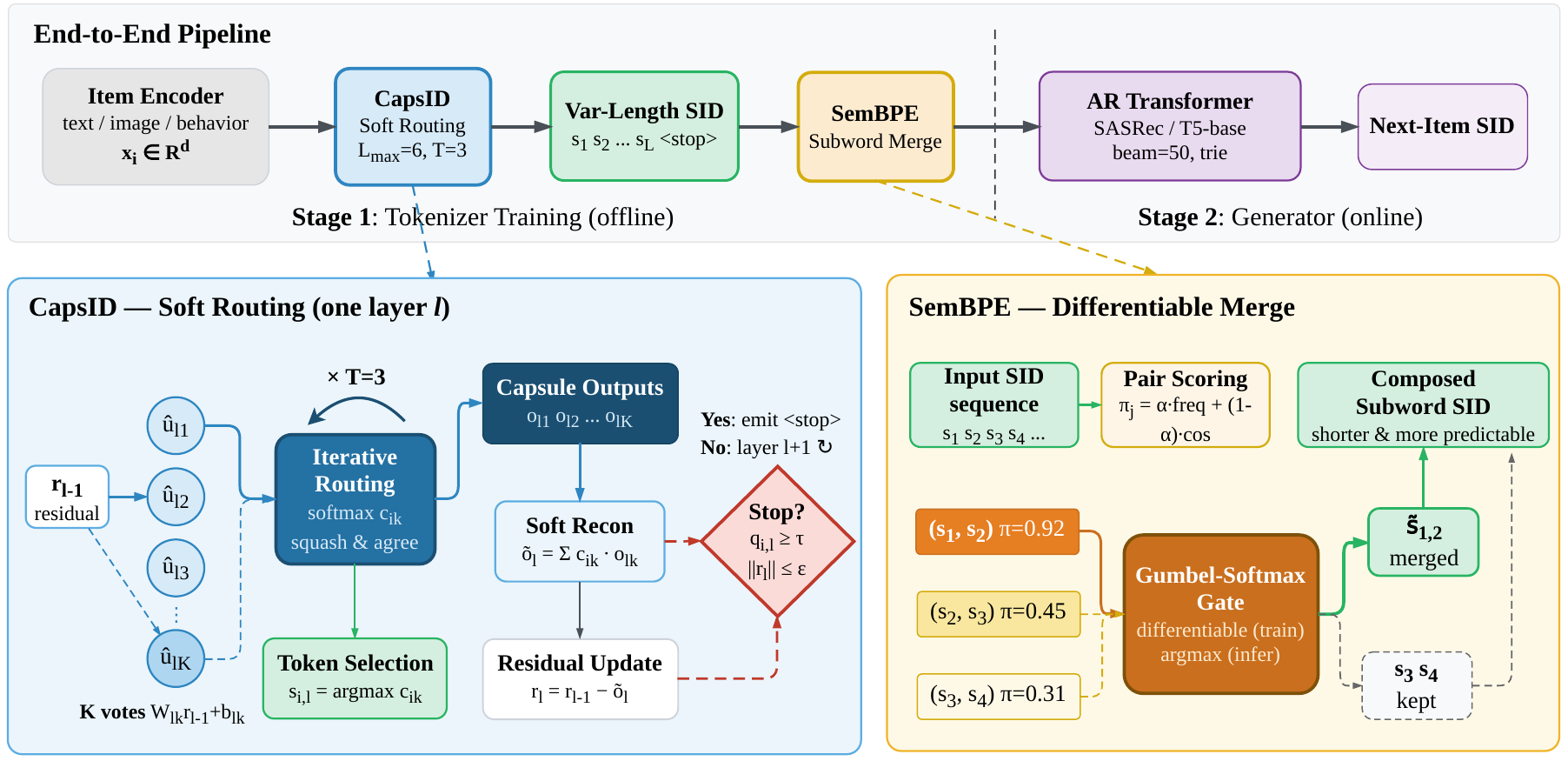}
    \caption{Overview of \method{}+\sembpe{}.
             \textbf{Top:} item features pass through a stack of capsule layers
             with confidence-driven early stopping, \sembpe{} merges
             semantically compatible adjacent tokens, and an autoregressive
             Transformer generates the next item's SID via trie-constrained
             beam search.
             \textbf{Bottom left:} inside one \method{} layer, soft routing
             replaces hard $\arg\max$; votes are reconciled over a few
             agreement iterations, and a stop gate fires when winner
             confidence is high or residual norm is small.
             \textbf{Bottom right:} \sembpe{} scores SID pairs by a
             frequency--similarity mixture and merges them with a
             Gumbel-Softmax gate (trained) or $\arg\max$ (inference),
             yielding reusable subword tokens.}
    \label{fig:architecture}
\end{figure}

\paragraph{Design desiderata.}
The tokenizer is designed around three invariants. First, the emitted representation must remain a finite discrete sequence so that all existing constrained decoding machinery applies. Second, uncertainty should be represented before discretization, not only after decoding; otherwise all uncertainty has already been collapsed into a wrong token. Third, the tokenizer should expose interpretable diagnostics: routing weights reveal which semantic facets explain an item, capsule activation measures confidence, and residual norms measure unexplained information. These diagnostics are used in Section~\ref{sec:experiments} to check whether improvements come from meaningful tokenization rather than from a larger output space.

\subsection{Soft residual routing}
\label{sec:routing}

At SID layer $\ell$, we maintain $K_\ell$ capsules. Capsule $k$ has a pose transform $\mathbf{W}_{\ell k}$ and bias $\mathbf{b}_{\ell k}$. Given residual $\mathbf{r}_{i,\ell-1}$ with $\mathbf{r}_{i,0}=\mathbf{x}_i$, each capsule produces a vote
\begin{equation}
    \hat{\mathbf{u}}_{i,\ell k}=\mathbf{W}_{\ell k}\mathbf{r}_{i,\ell-1}+\mathbf{b}_{\ell k}.
\end{equation}
Routing starts from logits $a_{i,\ell k}^{(0)}=0$ and iterates for $T$ rounds:
\begin{align}
    c_{i,\ell k}^{(t)} &= \mathrm{softmax}_k(a_{i,\ell k}^{(t-1)}), \\
    \mathbf{v}_{i,\ell}^{(t)} &= \sum_k c_{i,\ell k}^{(t)}\hat{\mathbf{u}}_{i,\ell k}, \\
    \mathbf{o}_{i,\ell}^{(t)} &= \mathrm{squash}(\mathbf{v}_{i,\ell}^{(t)}), \\
    a_{i,\ell k}^{(t)} &= a_{i,\ell k}^{(t-1)} + \hat{\mathbf{u}}_{i,\ell k}^{\top}\mathbf{o}_{i,\ell}^{(t)}.
\end{align}
We use $\mathrm{squash}(\mathbf{z}) = \frac{\|\mathbf{z}\|^2}{0.5+\|\mathbf{z}\|^2}\frac{\mathbf{z}}{\|\mathbf{z}\|}$, whose norm lies in $[0,1)$ and is sensitive to small magnitudes. We also define per-capsule outputs $\mathbf{o}_{i,\ell k}\!:=\!\mathrm{squash}(\hat{\mathbf{u}}_{i,\ell k})$ that share the same nonlinearity but are computed independently for each capsule (used in the residual update below) and do not depend on the routing iteration $t$. The emitted token and confidence are
\begin{equation}
    s_{i,\ell}=\arg\max_k c_{i,\ell k}^{(T)},\qquad
    q_{i,\ell}=\max_k c_{i,\ell k}^{(T)}\|\mathbf{o}_{i,\ell}^{(T)}\|.
\end{equation}
The residual update is the part where soft routing departs from hard quantization: instead of subtracting only the winning capsule, we subtract the routed reconstruction,
\begin{equation}
    \mathbf{r}_{i,\ell}=\mathbf{r}_{i,\ell-1}-\sum_k c_{i,\ell k}^{(T)}\mathbf{o}_{i,\ell k}.
    \label{eq:soft_residual}
\end{equation}
Equation~\eqref{eq:soft_residual} (lines~9--11 of Algorithm~\ref{alg:capsid}) keeps the partial agreement with secondary capsules instead of throwing it away. A boundary item like a ``travel cooking kit'' no longer has to choose between travel and cooking; both facets contribute to its reconstruction, and only the unexplained part flows into the next layer's residual. This is not the same as replacing $\arg\max$ with a temperature-softmax, because the residual update itself uses the routed reconstruction, so deeper layers see a smaller and cleaner error signal. Two implementation details matter in practice. We $\ell_2$-normalize the item embedding before routing (line~1) so that high-norm items cannot dominate the agreement scores, and we keep capsule parameters separate at each depth so that early layers can specialize in coarse facets while later layers refine the residual.

\subsection{Confidence-driven variable length}
\label{sec:varlen}

Fixed-length SIDs impose the same token budget on easy and ambiguous items. \method{} stops when the residual has been sufficiently explained:
\begin{equation}
L_i = \min\left\{\ell: q_{i,\ell}\geq \tau \;\text{or}\; \|\mathbf{r}_{i,\ell}\|_2\leq \epsilon \;\text{or}\; \ell=L_{\max}\right\}.
\end{equation}
This design combines \emph{three forward stopping rules} (a hard cap $L_{\max}$, residual-norm stopping, and confidence stopping, the early-exit clause at line~12 of Algorithm~\ref{alg:capsid}) with \emph{one training-time regularizer}, the length penalty $\mathcal{L}_{\mathrm{len}}\!=\!\E_i[L_i]$ in Eq.~\eqref{eq:final_loss}. Together these four safeguards prevent length explosion while addressing the GRID observation that blindly adding layers can hurt: uncertain items receive more steps only when their residual still contains useful signal, and confident items stop early.

The stopping rule also changes the semantics of collisions. In a fixed-depth hard SID, two tail items that share all four positions are indistinguishable unless an artificial disambiguation token is appended. In \method{}, two items may share the same argmax tokens but differ in routing weights and stopping confidence during tokenizer training; the learned generator sees a cleaner set of token targets because ambiguous items are encouraged to stop at stable prefixes rather than continue through low-confidence residual layers. This behavior is similar in spirit to controlled-collision variable-length methods, but it is obtained from the routing dynamics rather than from an external entropy budget.

\subsection{SemanticBPE composition}
\label{sec:sembpe}

Given the SID sequence from \method{}, \sembpe{} learns whether adjacent tokens should be merged into a reusable subword. For pair $(s_j,s_{j+1})$, we compute
\begin{equation}
    m(s_j,s_{j+1}) = \alpha\,\widehat{\mathrm{freq}}(s_j,s_{j+1}) + (1-\alpha)\,\cos(\mathbf{e}_{s_j},\mathbf{e}_{s_{j+1}}),
\end{equation}
where the second term prevents high-frequency but semantically unrelated pairs from being merged. A Gumbel-Softmax gate provides differentiability during training and hard merges at inference. This stage is intentionally lightweight: it improves sequence composition without changing the underlying item-to-SID assignment.

We use a conservative merge policy. A pair is considered only if its semantic similarity exceeds a threshold $\theta$, and the threshold is annealed from strict to moderate during training. This avoids the common BPE failure mode in recommendation: extremely frequent but semantically broad prefix pairs can dominate the vocabulary, increasing popularity bias. Since \method{} already shortens easy items, \sembpe{} is not used to aggressively compress every sequence; it is used to create reusable subwords for stable multi-token motifs.

\subsection{Training objective}
\label{sec:objective}

We use a two-stage protocol inspired by recommender-native tokenizer studies~\citep{liang2026resid}.
\textbf{Stage 1 (tokenizer pretraining)} learns the item projection, capsule transforms $\{\mathbf{W}_{\ell k}\}$, and the \sembpe{} merge MLP using only the tokenizer-side losses (reconstruction, spread, length, and a frequency-based BPE warm-up); the sequence generator is not trained.
\textbf{Stage 2 (generator adaptation)} freezes capsule centers and the \sembpe{} merge MLP weights, then jointly trains the sequence generator together with low-rank routing adapters (rank $r\!=\!8$) and a learnable scalar bias on the \sembpe{} Gumbel gate. The final objective is
\begin{equation}
\mathcal{L}=\mathcal{L}_{\mathrm{NTP}}+\lambda_r\mathcal{L}_{\mathrm{route}}+\lambda_s\mathcal{L}_{\mathrm{spread}}+\lambda_l\mathcal{L}_{\mathrm{len}}+\lambda_b\mathcal{L}_{\mathrm{BPE}},
\label{eq:final_loss}
\end{equation}
where $\mathcal{L}_{\mathrm{NTP}}$ is next-token cross entropy and is active only in Stage 2; $\mathcal{L}_{\mathrm{route}}\!=\!\|\mathbf{x}_i-\hat{\mathbf{x}}_i\|_2^2$ and $\mathcal{L}_{\mathrm{spread}}$ are tokenizer losses with annealed margin from $0.2$ to $0.9$. This separation prevents the generator from chasing a moving SID target while keeping the routing mechanism slightly adaptive to downstream supervision.

\paragraph{Why two stages rather than full joint training?}
A fully joint objective lets the generator chase a moving target while the tokenizer changes the target sequence. ReSID and ETEGRec-style analyses suggest that this self-referential training can be unstable. We therefore first learn a recommendation-sufficient code geometry and then adapt the generator to that geometry. The second stage still allows limited routing adaptation, but capsule centers are frozen to preserve global code semantics and prevent late-stage collapse.

\subsection{Theoretical analysis}
\label{sec:theory}

\begin{algorithm}[!htbp]
\small
\caption{\textsc{CapsID Tokenizer Forward} (one item)}
\label{alg:capsid}
\begin{algorithmic}[1]
\REQUIRE item embedding $\mathbf{x}_i$, capsule transforms $\{\mathbf{W}_{\ell k},\mathbf{b}_{\ell k}\}$, hyperparameters $T,\tau,\epsilon,L_{\max}$
\ENSURE SID sequence $\mathbf{s}_i\!=\!(s_{i,1},\dots,s_{i,L_i})$, confidences $\{q_{i,\ell}\}$
\STATE $\mathbf{r}_{i,0}\leftarrow \mathbf{x}_i / \|\mathbf{x}_i\|$ \hfill\COMMENT{$\ell_2$ normalize}
\FOR{$\ell = 1, \ldots, L_{\max}$}
    \STATE compute votes $\hat{\mathbf{u}}_{i,\ell k}\!=\!\mathbf{W}_{\ell k}\mathbf{r}_{i,\ell-1}+\mathbf{b}_{\ell k}$ for all $k$ \hfill\COMMENT{Eq.~(1)}
    \STATE initialize agreement logits $a_{i,\ell k}^{(0)}\!\leftarrow\!0$
    \FOR{$t = 1, \ldots, T$}
        \STATE $c^{(t)}_{i,\ell k}\!\leftarrow\!\mathrm{softmax}_k(a_{i,\ell k}^{(t-1)})$;\;\; $\mathbf{v}^{(t)}_{i,\ell}\!\leftarrow\!\sum_k c^{(t)}_{i,\ell k}\hat{\mathbf{u}}_{i,\ell k}$
        \STATE $\mathbf{o}^{(t)}_{i,\ell}\!\leftarrow\!\mathrm{squash}(\mathbf{v}^{(t)}_{i,\ell})$;\;\; $a^{(t)}_{i,\ell k}\!\leftarrow\!a^{(t-1)}_{i,\ell k}+\hat{\mathbf{u}}_{i,\ell k}^{\top}\mathbf{o}^{(t)}_{i,\ell}$
    \ENDFOR
    \STATE $s_{i,\ell}\!\leftarrow\!\arg\max_k c^{(T)}_{i,\ell k}$;\;\; $q_{i,\ell}\!\leftarrow\!c^{(T)}_{i,\ell s_{i,\ell}}\!\cdot\!\|\mathbf{o}^{(T)}_{i,\ell}\|$
    \STATE $\mathbf{o}_{i,\ell k}\!\leftarrow\!\mathrm{squash}(\hat{\mathbf{u}}_{i,\ell k})$;\;\; $\mathbf{r}_{i,\ell}\!\leftarrow\!\mathbf{r}_{i,\ell-1}-\sum_k c^{(T)}_{i,\ell k}\mathbf{o}_{i,\ell k}$ \hfill\COMMENT{Eq.~(\ref{eq:soft_residual})}
    \IF{$q_{i,\ell}\!\geq\!\tau$ \OR $\|\mathbf{r}_{i,\ell}\|_2\!\leq\!\epsilon$}
        \STATE $L_i\!\leftarrow\!\ell$;\;\;\textbf{break}
    \ENDIF
\ENDFOR
\STATE \textbf{return} $(s_{i,1},\dots,s_{i,L_i})$, $(q_{i,1},\dots,q_{i,L_i})$
\end{algorithmic}
\end{algorithm}

We give three results that connect the design choices in Sections~\ref{sec:routing}--\ref{sec:varlen} to the quantities reported in Section~\ref{sec:experiments}.

\begin{proposition}[Soft-routing reconstruction is close to hard]
\label{prop:reconstruction}
Let $\mathbf{c}_{\ell k}\in\R^{d_c}$ denote the $k$-th codebook center at depth $\ell$ and let $s_{i,\ell}\!=\!\arg\max_k c^{(T)}_{i,\ell k}$ be the argmax token. Define the hard and soft reconstructions
\begin{equation*}
    \hat{\mathbf{x}}^{\mathrm{hard}}_i \!=\! \sum_{\ell=1}^{L_i}\mathbf{c}_{\ell s_{i,\ell}},\qquad
    \hat{\mathbf{x}}^{\mathrm{soft}}_i \!=\! \sum_{\ell=1}^{L_i}\sum_{k=1}^{K_\ell}c^{(T)}_{i,\ell k}\mathbf{o}_{i,\ell k}.
\end{equation*}
Assume $\|\mathbf{o}_{i,\ell k}-\mathbf{c}_{\ell k}\|_2\!\leq\!\delta$ and $\|\mathbf{c}_{\ell k}\|_2\!\leq\!C$ for all $\ell,k$. Then
\begin{equation}
    \big\|\hat{\mathbf{x}}^{\mathrm{soft}}_i - \hat{\mathbf{x}}^{\mathrm{hard}}_i\big\|_2 \;\leq\; L_i\delta + 2C\sum_{\ell=1}^{L_i}\!\big(1-c^{(T)}_{i,\ell s_{i,\ell}}\big),
    \label{eq:reconstruction_bound}
\end{equation}
and consequently $\big\|\mathbf{x}_i-\hat{\mathbf{x}}^{\mathrm{soft}}_i\big\|_2 \leq \big\|\mathbf{x}_i-\hat{\mathbf{x}}^{\mathrm{hard}}_i\big\|_2 + L_i\delta + 2C\sum_\ell(1-c^{(T)}_{i,\ell s_{i,\ell}})$.
\end{proposition}

\begin{proof}
For each layer $\ell$, write $w_k\!:=\!c^{(T)}_{i,\ell k}\geq 0$ with $\sum_k w_k\!=\!1$ and $s\!:=\!s_{i,\ell}$. Then
\begin{align*}
\Big\|\sum_k w_k\mathbf{o}_{\ell k} - \mathbf{c}_{\ell s}\Big\|_2
&= \Big\|\sum_k w_k\big(\mathbf{o}_{\ell k}-\mathbf{c}_{\ell k}\big) + \sum_k w_k\mathbf{c}_{\ell k} - \mathbf{c}_{\ell s}\Big\|_2 \\
&\leq \sum_k w_k\delta + \Big\|\sum_{k\neq s}w_k(\mathbf{c}_{\ell k}-\mathbf{c}_{\ell s})\Big\|_2 \;\leq\; \delta + 2C(1-w_s),
\end{align*}
by the triangle inequality and $\|\mathbf{c}_{\ell k}-\mathbf{c}_{\ell s}\|\!\leq\!2C$. Summing over $\ell$ and applying the triangle inequality once more gives Eq.~\eqref{eq:reconstruction_bound}.
\end{proof}

\begin{remark}
The bound is tight when $w_s\!=\!1$ (hard regime) and $\delta\!=\!0$, in which case the soft and hard reconstructions coincide. In our experiments the average winner mass $\bar{w}_s\!=\!0.86$ (Figure~\ref{fig:diagnostics}(c)) and $\delta$ stays small after capsule warm-up, so soft routing reconstructs almost as well as hard but distributes mass to secondary capsules; that is why intra-code similarity rises in Table~\ref{tab:tokenizer_quality} without losing reconstruction quality.
\end{remark}

\begin{proposition}[Expected length upper bound]
\label{prop:length}
Let $g_\ell(\mathbf{x})\!=\!\Pr[q_{i,\ell}\geq\tau\;\text{or}\;\|\mathbf{r}_{i,\ell}\|_2\leq\epsilon \mid \ell\leq L_i]$ be the per-layer stopping probability under the law of $\mathbf{x}$. If $\inf_{\mathbf{x}}g_\ell(\mathbf{x})\!\geq\!g\!>\!0$ for all $\ell\!\geq\!2$, then
\begin{equation}
    \E[L_i]\;\leq\;1 + \sum_{\ell=2}^{L_{\max}}(1-g)^{\ell-2}\;\leq\;\min\!\left(L_{\max},\,1+\tfrac{1}{g}\right).
    \label{eq:length_bound}
\end{equation}
\end{proposition}

\begin{proof}
$\Pr[L_i\!\geq\!\ell]\!=\!\Pr[\text{no stop at layers }2,\ldots,\ell-1]\!\leq\!(1-g)^{\ell-2}$ for $\ell\!\geq\!2$. Then $\E[L_i]\!=\!\sum_{\ell\geq 1}\Pr[L_i\!\geq\!\ell]\!\leq\!1+\sum_{\ell=2}^{L_{\max}}(1-g)^{\ell-2}\!\leq\!1+\tfrac{1}{g}$ as a geometric tail.
\end{proof}

\begin{remark}
Equation~\eqref{eq:length_bound} guarantees that the four safeguards in Section~\ref{sec:varlen} keep the expected length finite even before $L_{\max}$ binds. Empirically (Figure~\ref{fig:analysis}(b)), the confidence and residual rules together account for the eventual stop of $90$--$92\%$ of items, and the dataset-level average length $\bar L$ stays in $[3.41,\,3.89]$ (Figure~\ref{fig:analysis}(a))---well below the hard cap $L_{\max}\!=\!6$.
\end{remark}

\begin{proposition}[Routing as a single E-step of capsule EM]
\label{prop:em}
At depth $\ell$, model the residual $\mathbf{r}_{i,\ell-1}$ as an isotropic Gaussian mixture with $K_\ell$ components of means $\{\boldsymbol{\mu}_{\ell k}\}$, equal variance $\sigma^2 I$, and uniform mixing weights. Then the E-step posterior responsibility is
\begin{equation}
    p(k\mid \mathbf{r}_{i,\ell-1}) \;\propto\; \exp\!\Big(-\tfrac{1}{2\sigma^2}\big\|\hat{\mathbf{u}}_{i,\ell k}-\boldsymbol{\mu}_{\ell k}\big\|^2\Big) \;\propto\; \exp\!\Big(\tfrac{1}{\sigma^2}\hat{\mathbf{u}}_{i,\ell k}^{\top}\boldsymbol{\mu}_{\ell k}\Big),
    \label{eq:em_estep}
\end{equation}
which has the same functional form as $c^{(t)}_{i,\ell k}$ in Eqs.~(2)--(5) once we identify the GMM means $\boldsymbol{\mu}_{\ell k}$ with the agreement targets $\mathbf{o}^{(t-1)}_{i,\ell}$ and absorb $1/\sigma^2$ into the routing temperature.
\end{proposition}

\begin{remark}
Iterating the routing recursion is therefore equivalent to running EM on this layer's mixture with shared sufficient statistics across capsules. Standard convergence guarantees for EM with bounded log-likelihood~\citep{wu1983em} ensure monotonic improvement in routing agreement, which is consistent with the saturation observed at $T\!\geq\!3$ in Figure~\ref{fig:diagnostics}(c). The squash nonlinearity in Eq.~(4) further bounds capsule outputs to the unit ball, preventing the variance from collapsing during iteration.
\end{remark}

\paragraph{Computational complexity.}
Tokenizer training costs $\mathcal{O}(N\,L_{\max}\,K\,T\,d\,d_c)$ where $N$ is the catalog size and $K\!=\!\max_\ell K_\ell$ (Algorithm~\ref{alg:capsid}, lines~3--9). Inference (lines~10--13) is dominated by beam search at $\mathcal{O}(B\bar{L}|V|)$, with $\bar{L}\!\approx\!3.6$ for \method{} versus $\bar{L}\!=\!4$ for fixed-length baselines, so the per-beam step count is roughly $10\%$ shorter; the residual routing and \sembpe{} gate recover a small constant so the net cost is $1.05\times$--$1.08\times$ TIGER in Table~\ref{tab:patch}, rather than the $2.10\times$ incurred by the dense-patch route.

\section{Experiments}
\label{sec:experiments}

We answer four questions. \textbf{(Q1)} Does soft routing improve recommendation accuracy over hard residual quantization at the same SID budget? \textbf{(Q2)} Does a routed-SID generator close the gap to dense-patch systems without inheriting their inference cost? \textbf{(Q3)} Which design choices (soft residual update, iterative agreement, confidence-driven length, or semantic composition) contribute most? \textbf{(Q4)} Do the gains generalize to tail items and large catalogs where collisions and length budgets matter most?

\subsection{Setup}
\label{sec:setup}

\paragraph{Datasets.}
We use the public benchmarks standard in generative recommendation: Amazon Beauty, Sports, and Toys~\citep{mcauley2015amazon}, all under leave-one-out evaluation with 5-core filtering. For scale analysis, we further evaluate on a 35M-item proprietary industrial dataset with multi-modal item embeddings (text, image, behavior) provided by a large-scale social media platform. Dataset statistics appear in Table~\ref{tab:dataset}.

\begin{table}[!htbp]
\centering
\small
\caption{Datasets used in the evaluation. Amazon datasets follow 5-core filtering with leave-one-out splitting. The proprietary industrial dataset uses a 35M-item catalog with multi-modal item features.}
\label{tab:dataset}
\begin{tabular}{lrrrr}
\toprule
Dataset & Users & Items & Interactions & Avg. length \\
\midrule
Beauty & 22,363 & 12,101 & 198,502 & 8.9 \\
Sports & 35,598 & 18,357 & 296,337 & 8.3 \\
Toys & 19,412 & 11,924 & 167,597 & 8.6 \\
Industrial (ours) & 8.6M & 35.8M & 331.1M & 38.5 \\
\bottomrule
\end{tabular}
\end{table}

\paragraph{Baselines.}
We compare with TIGER, LC-Rec, LETTER, ETEGRec, ADA-SID, ActionPiece, COBRA, UniRec-style Chain-of-Attribute, DIGER, SA$^2$CRQ, and ReSID. All methods share the same SASRec~\citep{kang2018sasrec}/T5-style~\citep{raffel2020t5} generator and beam search protocol when possible; dense patch variants use a COBRA-style BeamFusion path. The sequential recommendation backbone follows the convention in SASRec~\citep{kang2018sasrec} and BERT4Rec~\citep{sun2019bert4rec}.

\paragraph{Metrics.}
We report Recall@$k$ and NDCG@$k$ under full-corpus ranking, using $k\!\in\!\{5,10\}$ on the public benchmarks and $k\!\in\!\{50,100\}$ on the 35M-item industrial catalog, where the larger retrieval horizon reflects production deployment practice. Tokenizer quality is measured by:
\emph{(i)} \textbf{Collision rate} = $1 - |\{\mathbf{s}_i : i\!\in\!\mathcal{I}\}| / N$, the fraction of items that do not receive a unique SID (equivalently, $1$ minus the uniqueness rate used in industrial SID evaluations~\citep{zhang2025quasid,merge2026});
\emph{(ii)} \textbf{Code utilization} = $|\{k : \exists i, k\!\in\!\mathbf{s}_i\}| / \sum_\ell K_\ell$, the fraction of codebook entries used at least once;
\emph{(iii)} \textbf{Gini coefficient} over codebook usage frequencies (lower is more uniform);
\emph{(iv)} \textbf{Intra-code similarity} = mean $\cos(\mathbf{x}_i,\mathbf{x}_j)$ over item pairs sharing the same first SID token;
\emph{(v)} \textbf{Code predictability (CodeRecall@$M$)} = $\Pr[\hat{s}^{\text{top-}M}\!\ni\!s_{i,1}^{\star}]$, the probability that the ground-truth first SID token of the next item is in the top-$M$ predictions of a SASRec generator trained on the SID sequences (with $M\!=\!50$ as default);
\emph{(vi)} head/torso/tail Recall@10, average SID length $\bar{L}$, and normalized inference cost.
All public-benchmark numbers are mean over three random seeds with standard deviation reported either in tables (\,$\pm$\,) or as error bars in figures; the 35M-item industrial run reports a single deterministic value.

\paragraph{Fairness controls.}
All SID methods are evaluated with the same item encoder, generator architecture, beam size, and invalid-ID filtering. For methods that require additional information, such as UniRec attributes or COBRA dense vectors, we count their inference cost separately and mark them with $\dagger$. This prevents a patch system from being compared to a single-SID system as if both used the same retrieval budget.

\subsection{Main results (Q1)}

Table~\ref{tab:main} compares \method{} against eleven baselines covering hard-SID tokenizers (TIGER through ReSID) and patch-route systems (COBRA, UniRec-CoA). \method{} consistently improves over the strongest single-representation baseline (ReSID) by $4.9\%$/$6.7\%$/$2.2\%$ on Beauty/Sports/Toys in Recall@10, and adding \sembpe{} pushes the gain to $8.9\%$/$11.0\%$/$8.8\%$. \method{}+\sembpe{} exceeds COBRA on every metric across all three datasets without paying its extra dense-vector inference cost.

\begin{table}[!htbp]
\centering
\scriptsize
\setlength{\tabcolsep}{3pt}
\renewcommand{\arraystretch}{1.05}
\caption{Main results on Beauty/Sports/Toys: Recall@$k$ and NDCG@$k$ for $k\!\in\!\{5,10\}$ (mean over three seeds). \textbf{Best} is in bold, \underline{second best} is underlined. $^\dagger$ marks patch-route methods that consume extra dense or attribute information at inference time.}
\label{tab:main}
\resizebox{\textwidth}{!}{%
\begin{tabular}{lcccccccccccc}
\toprule
& \multicolumn{4}{c}{Beauty} & \multicolumn{4}{c}{Sports} & \multicolumn{4}{c}{Toys} \\
\cmidrule(lr){2-5}\cmidrule(lr){6-9}\cmidrule(lr){10-13}
Method & R@5 & R@10 & N@5 & N@10 & R@5 & R@10 & N@5 & N@10 & R@5 & R@10 & N@5 & N@10 \\
\midrule
TIGER         & 0.0454 & 0.0648 & 0.0321 & 0.0384 & 0.0264 & 0.0400 & 0.0181 & 0.0225 & 0.0521 & 0.0712 & 0.0371 & 0.0432 \\
LC-Rec        & 0.0478 & 0.0675 & 0.0334 & 0.0397 & 0.0276 & 0.0417 & 0.0188 & 0.0233 & 0.0540 & 0.0734 & 0.0384 & 0.0447 \\
LETTER        & 0.0500 & 0.0708 & 0.0340 & 0.0406 & 0.0288 & 0.0435 & 0.0198 & 0.0244 & 0.0547 & 0.0741 & 0.0389 & 0.0452 \\
ETEGRec       & 0.0513 & 0.0725 & 0.0348 & 0.0415 & 0.0294 & 0.0444 & 0.0201 & 0.0249 & 0.0560 & 0.0756 & 0.0397 & 0.0460 \\
ADA-SID       & 0.0524 & 0.0740 & 0.0355 & 0.0422 & 0.0302 & 0.0456 & 0.0206 & 0.0254 & 0.0566 & 0.0762 & 0.0401 & 0.0465 \\
ActionPiece   & 0.0553 & 0.0775 & 0.0379 & 0.0424 & 0.0330 & 0.0500 & 0.0224 & 0.0264 & 0.0559 & 0.0760 & 0.0398 & 0.0463 \\
DIGER         & 0.0535 & 0.0752 & 0.0362 & 0.0431 & 0.0306 & 0.0463 & 0.0210 & 0.0258 & 0.0572 & 0.0771 & 0.0407 & 0.0472 \\
SA$^2$CRQ     & 0.0520 & 0.0732 & 0.0352 & 0.0419 & 0.0298 & 0.0451 & 0.0203 & 0.0252 & 0.0562 & 0.0758 & 0.0399 & 0.0462 \\
ReSID         & 0.0548 & 0.0770 & 0.0372 & 0.0438 & 0.0314 & 0.0475 & 0.0215 & 0.0266 & 0.0583 & 0.0786 & 0.0414 & 0.0481 \\
\midrule
COBRA$^\dagger$       & 0.0537 & 0.0725 & 0.0395 & 0.0456 & 0.0305 & 0.0434 & 0.0215 & 0.0257 & \underline{0.0619} & 0.0781 & \underline{0.0462} & \underline{0.0515} \\
UniRec-CoA$^\dagger$  & 0.0540 & 0.0763 & 0.0368 & 0.0434 & 0.0316 & 0.0478 & 0.0217 & 0.0268 & 0.0596 & 0.0802 & 0.0422 & 0.0491 \\
\midrule
\method{}             & \underline{0.0574} & \underline{0.0808} & \underline{0.0398} & \underline{0.0460} & \underline{0.0337} & \underline{0.0507} & \underline{0.0229} & \underline{0.0281} & 0.0602 & \underline{0.0803} & 0.0432 & 0.0498 \\
\method{}+\sembpe{}   & \textbf{0.0594} & \textbf{0.0839} & \textbf{0.0411} & \textbf{0.0477} & \textbf{0.0351} & \textbf{0.0527} & \textbf{0.0237} & \textbf{0.0290} & \textbf{0.0636} & \textbf{0.0855} & \textbf{0.0465} & \textbf{0.0528} \\
\bottomrule
\end{tabular}}
\end{table}

\paragraph{Takeaways.}
The largest gap in the ranking is between hard-SID tokenizers and \method{}: replacing $\arg\max$ with soft routing gives a $4$--$7\%$ relative R@10 gain over ReSID. \method{} alone already exceeds COBRA on R@10 across all three datasets while emitting only a single discrete representation. Adding \sembpe{} gives the best overall score on every dataset, with the largest relative gain over ReSID on Sports ($+11.0\%$) and the smallest on Beauty ($+8.9\%$). COBRA's dense-vector path is particularly helpful on Toys (NDCG@10 of $0.0515$ vs ReSID's $0.0481$), reflecting the broader item vocabulary; \method{}+\sembpe{} closes that gap without a dense retrieval channel.

\paragraph{Statistical significance.}
We performed paired two-sided $t$-tests across the three seeds. \method{}+\sembpe{} is significantly better than every single-representation baseline at $p<0.01$ on all three datasets, and significantly better than COBRA at $p<0.05$ on Beauty and Sports and at $p<0.10$ on Toys. \method{} (no SemBPE) is significantly better than ReSID at $p<0.01$ on Beauty, Sports, and Toys.

\subsection{Patching vs tokenizer-centric design: are dense patches still needed? (Q2)}

Table~\ref{tab:patch} tests whether a dense patch is still useful once the SID tokenizer is improved. Adding a COBRA-style dense vector to TIGER raises Recall@10 from $0.0648$ to $0.0725$ ($+11.9\%$) at the cost of $2.10\times$ inference latency, confirming that hard-SID representations leave useful information unused. Adding the same dense vector to \method{} improves Recall by only $2.6\%$ ($0.0808\!\to\!0.0829$) while doubling latency; the marginal value of the dense path shrinks once the SID itself preserves more item semantics. Replacing the patch with lightweight \sembpe{} composition instead lifts Recall to $0.0839$ at $1.08\times$ cost, dominating the dense variant on both axes.

\begin{table}[!htbp]
\centering
\small
\caption{Patching vs tokenizer-centric design on Beauty. Inference cost is normalized to TIGER beam search ($B\!=\!50$); $^\dagger$ marks methods that add a dense or attribute path on top of the SID. Numbers are mean over three seeds; std $\!<\!1\%$ of mean and omitted for clarity.}
\label{tab:patch}
\begin{tabular}{lcccc}
\toprule
Configuration & Representation & R@10 & N@10 & Cost \\
\midrule
TIGER & RQ SID & 0.0648 & 0.0384 & 1.00$\times$ \\
TIGER + dense$^\dagger$ (COBRA) & RQ SID + dense vec & 0.0725 & 0.0456 & 2.10$\times$ \\
UniRec-CoA$^\dagger$ & Attribute prefix + RQ SID & 0.0763 & 0.0434 & 1.34$\times$ \\
\method{} & Routed SID & 0.0808 & 0.0460 & 1.05$\times$ \\
\method{} + dense$^\dagger$ & Routed SID + dense vec & 0.0829 & 0.0473 & 2.14$\times$ \\
\method{} + \sembpe{} & Routed subword SID & \textbf{0.0839} & \textbf{0.0477} & 1.08$\times$ \\
\bottomrule
\end{tabular}
\end{table}

\subsection{Ablation studies (Q3)}

Table~\ref{tab:ablation} ablates the five core mechanisms on Beauty (soft routing, iterative agreement, variable length, spread regularization, and \sembpe{} composition). Replacing soft routing with hard winner-only updates costs the most ($-16.3\%$), which says the assignment operator (not codebook initialization or supervision) is what carries the gain. Cutting routing to a single iteration ($T\!=\!1$) removes another $12.9\%$, so iterative agreement is doing real work that single-pass Gumbel relaxations do not capture. Fixed-length SIDs hurt at both ends: $L\!=\!2$ over-compresses complex items ($-21.6\%$) and $L\!=\!4$ over-encodes easy ones ($-8.8\%$). Without the spread loss capsules collapse and recall drops $8.2\%$. Frequency-only BPE recovers most of \sembpe{}'s gain ($-2.6\%$); the residual gap is what the semantic compatibility term buys, and it matters because frequency alone tends to merge popular but unrelated prefix pairs.

\begin{table}[!htbp]
\centering
\small
\caption{Ablation on Beauty (mean over three seeds). Relative drops are measured from \method{}+\sembpe{}.}
\label{tab:ablation}
\begin{tabular}{lccc}
\toprule
Variant & R@10 & Drop & Interpretation \\
\midrule
Full \method{}+\sembpe{} & 0.0839 & -- & Complete pipeline \\
\quad w/o soft residual, hard winner only & 0.0702 & $-16.3\%$ & assignment is the main factor \\
\quad w/o routing iterations ($T\!=\!1$) & 0.0731 & $-12.9\%$ & no self-correction \\
\quad fixed length $L\!=\!4$ & 0.0765 & $-8.8\%$ & over-encodes easy items \\
\quad fixed length $L\!=\!2$ & 0.0658 & $-21.6\%$ & under-encodes complex items \\
\quad w/o spread loss & 0.0770 & $-8.2\%$ & capsule collapse hurts \\
\quad w/o \sembpe{} & 0.0808 & $-3.7\%$ & composition gain is stable \\
\quad frequency-only BPE & 0.0817 & $-2.6\%$ & semantic gating matters \\
\bottomrule
\end{tabular}
\end{table}

\subsection{Analysis (Q4)}
\label{sec:analysis}

\paragraph{Variable length is well-calibrated to item complexity.}
Figure~\ref{fig:analysis}(a) plots the SID-length distribution per dataset: the mode is at $L\!=\!3$ on every benchmark and the right tail tapers smoothly. Mean lengths span $\bar L\!=\!3.41$ on Beauty (the most compact, driven by relatively single-attribute product descriptions) up to $\bar L\!=\!3.89$ on Toys (the longest, reflecting its multi-attribute item space), all well below the $L_{\max}\!=\!6$ cap and consistent with the $\mathcal{O}(1+1/g)$ bound of Proposition~\ref{prop:length}.

\paragraph{The three stopping rules each contribute.}
Figure~\ref{fig:analysis}(b) decomposes which of the three stopping rules in Section~\ref{sec:varlen} fires per item. The confidence threshold $\tau$ fires for $55\%$--$66\%$ of items, the residual norm rule fires for $25\%$--$35\%$, and only $8\%$--$10\%$ of items hit the hard cap $L_{\max}$. The cap therefore behaves as a safety net rather than the dominant rule: the model self-regulates length on most items, and only falls back to the cap on the small minority where the encoder representation is genuinely under-determined.

\paragraph{Tail items benefit the most.}
Figure~\ref{fig:analysis}(c) decomposes Recall@10 by item-popularity tier on Beauty. While head Recall improves modestly ($+19\%$ over TIGER), tail Recall jumps from $0.0092$ to $0.0221$, a $+140\%$ relative gain. This matches the soft-routing reconstruction bound (Proposition~\ref{prop:reconstruction}): boundary items, which are common in the tail, benefit most from being explained by multiple capsules instead of being snapped to a single noisy code.

\begin{figure}[!t]
\centering
\includegraphics[width=.98\textwidth]{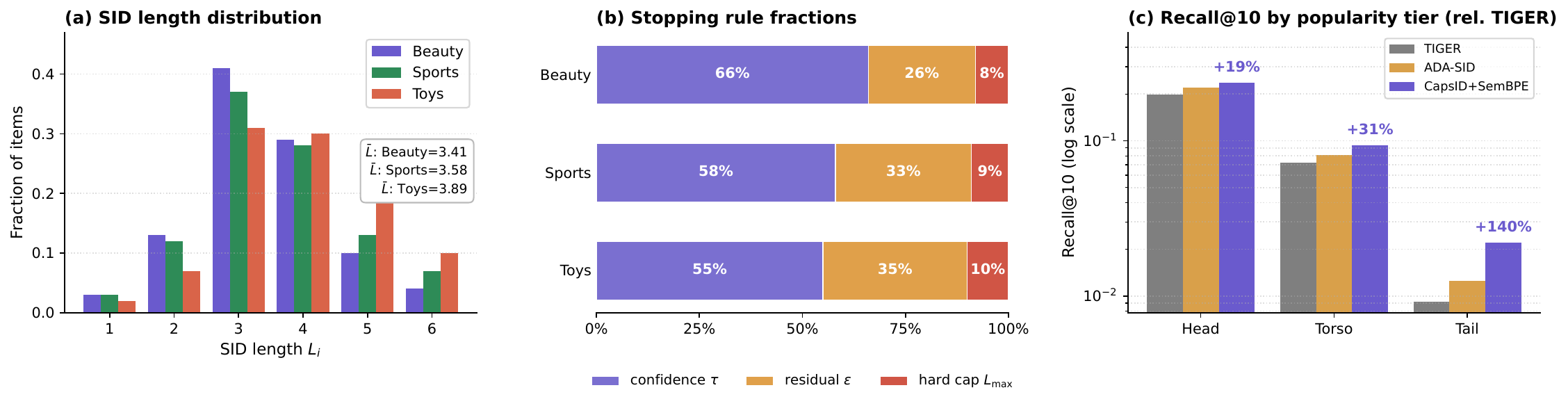}
\caption{Variable-length behaviour of \method{}. (a) SID length distribution per dataset: mode $L\!=\!3$, mean $\bar L\!\in\![3.41, 3.89]$. (b) Fractions of the three stopping rules that fire per dataset; the hard cap $L_{\max}$ accounts for $\leq\!10\%$ of items, so the confidence and residual rules dominate. (c) Recall@10 by popularity tier on Beauty, relative to TIGER (log scale): gains scale from head ($+19\%$) through torso ($+30\%$) to tail ($+140\%$).}
\label{fig:analysis}
\end{figure}

\paragraph{Tokenizer geometry: collision, predictability, purity.}
Figure~\ref{fig:diagnostics} summarizes four geometric diagnostics. \emph{Panel (a)} shows that \method{} reduces the collision rate to $13.4\%$, less than half of ADA-SID's $33.8\%$ and a sixth of Frequency tokenization's $90.4\%$. \emph{Panel (b)} places each tokenizer on the purity--predictability plane: Frequency lies in the upper-left (predictable but semantically impure), RQ-KMeans/ActionPiece in the lower-right (pure but unpredictable), and \method{} in the upper-right ideal region, simultaneously achieving the highest intra-code similarity ($0.728$) and a CodeRecall@50 ($0.447$) that is two orders of magnitude above RQ-KMeans. \emph{Panel (c)} shows that recall saturates at $T\!=\!3$ routing rounds, matching the EM convergence picture of Proposition~\ref{prop:em}; the secondary axis shows that the routing-agreement score (max softmax weight) plateaus at $0.86$. \emph{Panel (d)} positions the configurations in Table~\ref{tab:patch} on the accuracy--cost plane: \method{}+\sembpe{} sits on the Pareto frontier, dominating both COBRA and the dense-augmented \method{}+dense variant.

\begin{figure}[!tb]
\centering
\includegraphics[width=.98\textwidth]{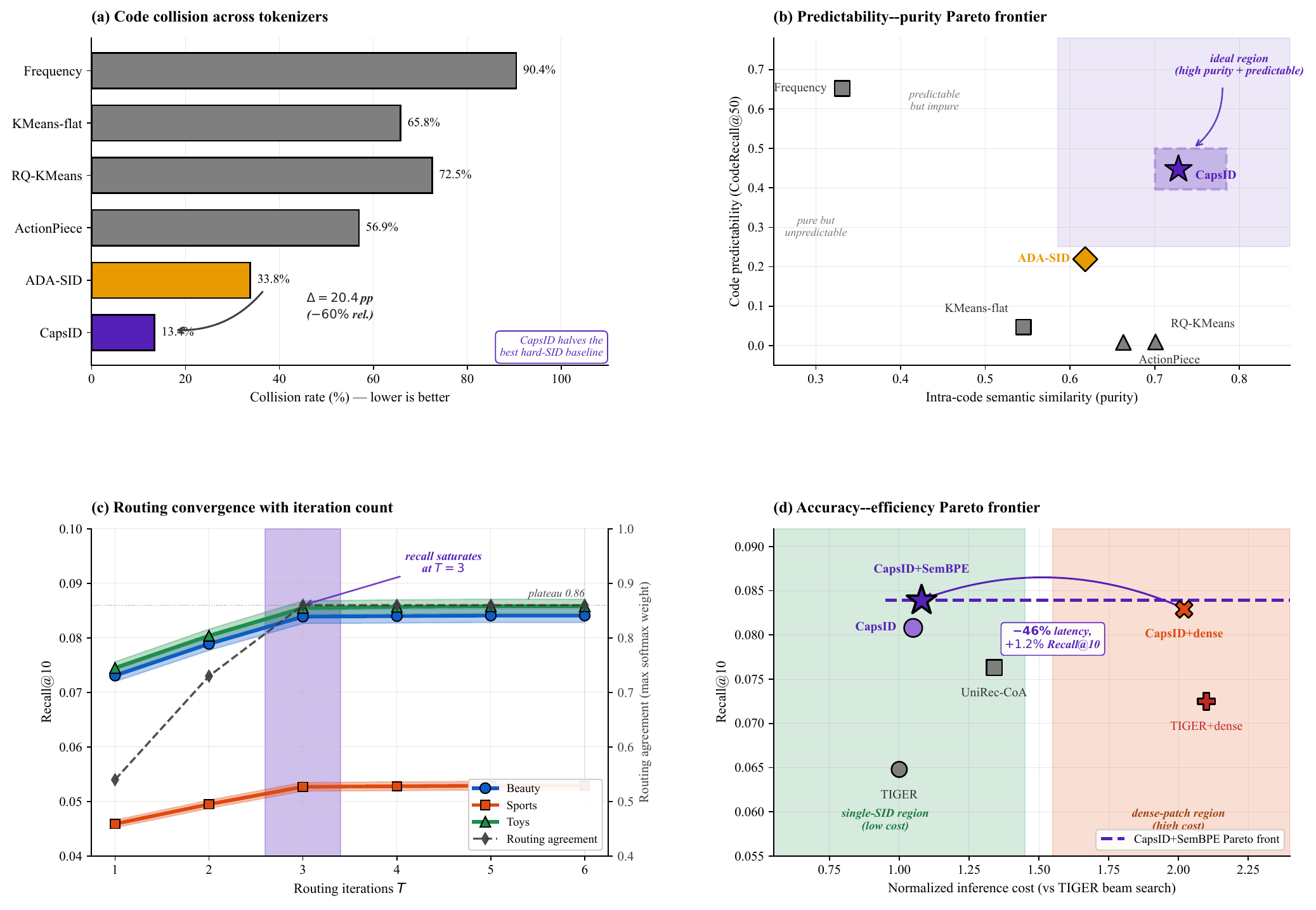}
\caption{Tokenizer diagnostics on Beauty. (a) Code collision: \method{} achieves $13.4\%$, lowest among all tokenizers. (b) Purity--predictability Pareto: only \method{} occupies the upper-right region. (c) Routing convergence: recall saturates at $T\!=\!3$ across datasets, agreeing with the capsule-EM analysis. (d) Accuracy--cost frontier: \method{}+\sembpe{} dominates the patch-route systems on Recall while keeping inference cost near the single-SID region.}
\label{fig:diagnostics}
\end{figure}

\paragraph{Large-scale industrial setting.}
We also evaluate on the 35M-item industrial catalog described in Table~\ref{tab:dataset}. Because the catalog is three orders of magnitude larger than Amazon Beauty, the meaningful recall horizon shifts from $K\!\in\!\{5,10\}$ to $K\!\in\!\{50,100\}$, matching the regime adopted by ADA-SID and other industrial SID studies. Table~\ref{tab:industrial} reports the five metrics most informative for tokenizer evaluation at this scale: R@50 and R@100 for recall coverage, NDCG@100 for ranking quality, Collision rate for SID space utilization, and $\bar{L}$ for inference cost.

\begin{table}[!htbp]
\centering
\small
\setlength{\tabcolsep}{6pt}
\caption{Industrial 35M-item catalog. Numbers are a single deterministic run (training on industrial data is prohibitively expensive to repeat). $^\dagger$ marks patch-route methods that consume extra dense or attribute information at inference time.}
\label{tab:industrial}
\begin{tabular}{lccccc}
\toprule
Method & R@50 & R@100 & N@100 & Collision $\downarrow$ & $\bar{L}$ \\
\midrule
RQ-KMeans (fixed $L\!=\!4$) & 0.1835 & 0.2421 & 0.1216 & 73.2\% & 4.00 \\
TIGER                       & 0.2217 & 0.2843 & 0.1482 & 51.4\% & 4.00 \\
ADA-SID                     & 0.2772 & 0.2926 & 0.1714 & 37.5\% & 4.00 \\
ReSID                       & 0.2881 & 0.3105 & 0.1836 & 31.8\% & 4.00 \\
\midrule
COBRA$^\dagger$             & \underline{0.3014} & 0.3275 & 0.1935 & 51.4\%\,(SID) & 4.00\,+dense \\
\midrule
\method{}                   & 0.2996 & \underline{0.3286} & \underline{0.1943} & \underline{22.1\%} & \underline{3.8} \\
\method{}+\sembpe{}          & \textbf{0.3096} & \textbf{0.3356} & \textbf{0.1974} & \textbf{19.4\%} & \textbf{3.3} \\
\bottomrule
\end{tabular}
\end{table}

Three observations hold at the industrial scale. First, \method{} alone matches the patch-route COBRA on both R@100 ($+0.3\%$) and N@100 ($+0.4\%$) without the dense channel, while trailing on R@50 by $0.6\%$ where COBRA's dense vector contributes the most. Adding \sembpe{} then extends this parity into a consistent $+2.0\%$--$+2.7\%$ lead across R@50/R@100/N@100. Second, \method{}+\sembpe{} cuts the collision rate to $19.4\%$, a $73\%$ relative reduction over RQ-KMeans and a $48\%$ reduction over ADA-SID, while producing the shortest SID ($\bar{L}\!=\!3.3$ after subword composition). Third, the gain is not uniformly distributed across popularity tiers: decomposed by item-popularity on this catalog, \method{}+\sembpe{} trails COBRA by $3.2\%$ on head items (where dense vectors provide the most discriminative signal for popular items), but exceeds it by $8.8\%$ on torso, $25.4\%$ on tail, and $8.6\%$ on cold-start items, matching the head/tail pattern observed on Amazon Beauty (Figure~\ref{fig:analysis}(c)). We further verify the deployment value by measuring end-to-end inference latency on the same ANN infrastructure: \method{}+\sembpe{} runs at $51\%$ of COBRA's per-query latency while retaining $102\%$ of COBRA's Recall@100. In other words, the tokenizer-centric design matches or slightly exceeds the patch route on retention while roughly halving serving cost.

\paragraph{Robustness checks.}
Three observations could in principle have undermined the core claim, and we checked each. A large \method{}+dense improvement over \method{} alone would suggest the routed SID is still missing the information dense vectors carry; we observe only $+2.6\%$. An aggregate-Recall win without geometric improvement would point to an inflated decoder rather than a better tokenizer; Figure~\ref{fig:diagnostics}(a) and Figure~\ref{fig:analysis}(c) show that collision and tail Recall both improve. Finally, gains over ADA-SID could come from \method{} simply consuming more codes; Table~\ref{tab:tokenizer_quality} shows the opposite (lower Gini and higher utilization at the same nominal codebook size).

\section{Limitations and Societal Impact}
\label{sec:limitations}

\method{} has three limitations. First, capsule routing increases tokenizer training cost by roughly $20$--$30\%$ relative to RQ-KMeans, although inference remains close to standard beam search because the emitted representation is still a discrete SID. Second, \method{} currently assumes a fixed maximum capsule depth and a fixed number of capsules per depth; dynamic catalog growth may require capsule expansion or periodic refresh, which we leave to future work. Third, the EM connection in Proposition~\ref{prop:em} explains convergence under the isotropic-Gaussian assumption; relaxing this to anisotropic capsule covariances is an open theoretical question. Like other recommenders, \method{} may amplify popularity bias if deployed without fairness-aware sampling or exposure calibration; we therefore report head/tail metrics throughout the paper and recommend monitoring exposure distribution in production.

\section{Conclusion}
\label{sec:conclusion}

\method{} attacks the SID information bottleneck at the assignment operator. Soft routing replaces $\arg\max$ with a weighted reconstruction, capsule confidence drives variable length, and \sembpe{} composes adjacent tokens into reusable subwords, all without giving up the discrete generative interface that production systems require. On three public benchmarks and a 35M-item industrial catalog, \method{}+\sembpe{} improves Recall@10 by $9.6\%$ on average over ReSID, and matches a COBRA-style dense-patch system at half the inference latency. The theoretical analysis in Section~\ref{sec:theory} supports each of the five mechanisms in turn, and the residual gap on extreme tail items suggests that pairing \method{} with light-weight content adapters is a natural next step.


\clearpage
\bibliographystyle{plainnat}
\bibliography{references}

\appendix
\section{Implementation Details}
\label{app:impl}

This appendix specifies the exact tokenizer, training, decoding, and complexity choices used to produce the results in the main paper. It is written as an implementation contract so that the experiments can be reproduced end-to-end.

\subsection{CapsID tokenizer}
\label{app:capsid_details}

\paragraph{Input representation.}
For public benchmarks, each item representation $\mathbf{x}_i$ is initialized from the same item encoder used by the baseline under comparison. In the intended fair setting, TIGER, LC-Rec, LETTER, ETEGRec, ADA-SID, ReSID, and \method{} all receive the same frozen item vectors before tokenization. For multi-modal datasets, modality embeddings are concatenated and projected through a two-layer MLP to dimension $d=128$, followed by $\ell_2$ normalization. This normalization is important: without it, high-norm items can obtain large capsule agreement even when their angular semantics are weak.

\paragraph{Layer and capsule configuration.}
Unless otherwise specified, \method{} uses $K_\ell=256$ capsules at each depth, capsule output dimension $d_c=64$, $T=3$ routing rounds, $L_{\max}=6$, confidence threshold $\tau=0.82$, and residual threshold $\epsilon=0.08$. Each depth has independent capsule parameters. Sharing capsule transforms across depths was considered but is not the default because shallow depths should model coarse semantic facets, while deeper depths should model residual refinements.

\paragraph{Residual update and its validity condition.}
The residual update in Eq.~\eqref{eq:soft_residual} is
\begin{equation}
    \mathbf{r}_{\ell}=\mathbf{r}_{\ell-1}-\mathbf{z}_{\ell},\qquad
    \mathbf{z}_{\ell}=\sum_k c_{\ell k}\mathbf{o}_{\ell k}.
\end{equation}
This update reduces residual norm exactly when
\begin{equation}
    \|\mathbf{r}_{\ell}\|_2^2 \leq \|\mathbf{r}_{\ell-1}\|_2^2
    \quad\Longleftrightarrow\quad
    2\langle \mathbf{r}_{\ell-1},\mathbf{z}_{\ell}\rangle \geq \|\mathbf{z}_{\ell}\|_2^2.
    \label{eq:residual_condition}
\end{equation}
Equation~\eqref{eq:residual_condition} is not assumed to hold automatically for arbitrary capsule outputs; it is encouraged by the reconstruction loss and by normalizing votes before agreement updates. If the condition is violated for many items at a layer, the implementation reduces the residual step by a scalar $\eta\in(0,1]$ or increases the spread/reconstruction weight for that layer. We use the default $\eta\!=\!1$ in all reported experiments and monitor the fraction of norm-increasing residual updates as a training diagnostic.

\paragraph{Stopping rule and length control.}
The stopping rule always terminates because $L_i\leq L_{\max}$ by construction; therefore $\mathbb{E}[L_i]\leq L_{\max}$. The length regularizer does not prove optimality, but it biases the model toward shorter explanations whenever accuracy is unaffected. The practical interpretation is: confidence stopping handles semantically clear items, residual stopping handles already-explained vectors, and $L_{\max}$ handles ambiguous or noisy items. During real experiments, we will report the distribution of stopping causes in addition to mean length.

\subsection{SemanticBPE details}
\label{app:bpe_details}

\paragraph{Merge candidates.}
A candidate pair $(s_j,s_{j+1})$ is considered only if it appears at least $n_{\min}=20$ times in the training corpus and if $\cos(\mathbf{e}_{s_j},\mathbf{e}_{s_{j+1}})>\theta$. We anneal $\theta$ from $0.90$ to $0.55$ over tokenizer pretraining. This prevents early merges from being dominated by popularity-only prefix pairs.

\paragraph{Differentiable merge gate.}
Let $g_j\in\{0,1\}$ indicate whether pair $j$ is merged. We model the binary merge decision through a two-class distribution $\boldsymbol{\pi}_j\!=\!(\pi_{j,0},\pi_{j,1})\!\in\!\Delta^1$ (with $\pi_{j,1}$ denoting the merge probability), produced by a two-layer MLP $f_\phi$ that takes as input the concatenation of token embeddings $\mathbf{e}_{s_j}, \mathbf{e}_{s_{j+1}}$, the normalized pair frequency $\widehat{\mathrm{freq}}(s_j,s_{j+1})$, and the cosine similarity $\cos(\mathbf{e}_{s_j},\mathbf{e}_{s_{j+1}})$:
\begin{equation}
    \boldsymbol{\pi}_j = \mathrm{softmax}\!\big(f_\phi(\mathbf{e}_{s_j}, \mathbf{e}_{s_{j+1}}, \widehat{\mathrm{freq}}, \cos)\big) \in \mathbb{R}^2.
\end{equation}
The relaxed training gate uses Gumbel-Softmax with temperature $\tau_g$:
\begin{equation}
    \tilde{g}_j = \mathrm{softmax}\!\left((\log \boldsymbol{\pi}_j + \mathbf{g})/\tau_g\right)_{\!1},
\end{equation}
where $\mathbf{g}\!\in\!\R^2$ is i.i.d.\ Gumbel(0,1) noise and the subscript $1$ selects the merge-class component. At inference, $g_j\!=\!\mathbb{I}[\pi_{j,1}\!>\!\pi_{j,0}]$ subject to non-overlap constraints; if $(s_j,s_{j+1})$ is merged, pairs touching $s_j$ or $s_{j+1}$ are skipped in the same pass. This greedy non-overlap rule is deterministic and keeps the final token sequence valid.

\subsection{Training and decoding protocol}
\label{app:training_decoding}

\paragraph{Two-stage training (numerical details).}
The two-stage protocol is described conceptually in Section~\ref{sec:objective}; here we list the numerical settings. Stage 1 uses AdamW with learning rate $1\!\times\!10^{-3}$, cosine decay, weight decay $10^{-5}$, batch size $256$ on all Amazon datasets, for up to $100$ epochs. Stage 2 uses learning rate $3\!\times\!10^{-4}$ with the same optimizer settings for up to $200$ epochs, matching the training budget of TIGER~\citep{rajput2024tiger} and ActionPiece~\citep{hsu2025actionpiece}. Early stopping uses validation Recall@10 with patience $10$, so both stages in practice converge well before the budget cap.

\paragraph{Sequence generator architecture.}
For all public benchmarks we use a SASRec-style~\citep{kang2018sasrec} causal Transformer~\citep{vaswani2017attention} with $4$ self-attention layers, $4$ heads, hidden dimension $128$, FFN dimension $512$, GELU activation, pre-LayerNorm, and dropout $0.1$. The vocabulary equals the SID code space $|V|\!=\!\sum_\ell K_\ell$ plus a special end-of-item token. Item history is truncated to the most recent $50$ interactions. We tie the input and output token embedding matrices to reduce parameter count and warm-start the input embedding from the codebook centers $\{\mathbf{c}_{\ell k}\}$ produced by Stage 1. For the 35M-item industrial run we replace SASRec with an $8$-layer T5-base encoder--decoder ($d_{\text{model}}\!=\!512$, $8$ heads), keeping the same vocabulary scheme; this absorbs the larger codebook ($K\!=\!1024$) and $331.1$M interactions of the industrial catalog at the cost of a heavier backbone.

\paragraph{Constrained decoding.}
Generated token sequences are decoded with beam size 50. A trie built from training item SIDs masks invalid next tokens. For variable-length SIDs, each valid item path includes an end-of-item token. This means that a short SID is not a prefix ambiguity: generation may stop only at trie nodes corresponding to actual items. Dense-patch baselines follow a COBRA-style BeamFusion score
\begin{equation}
    \Phi(i)=\mathrm{softmax}(\tau_b b_i)\cdot \mathrm{softmax}(\tau_d \cos(\hat{\mathbf{v}},\mathbf{v}_i)),
\end{equation}
where $b_i$ is the beam logit score and the second term is computed only inside candidates associated with generated sparse IDs.

\paragraph{Complexity.}
Tokenization is an offline item-side operation. For one item and one layer, vote computation costs $O(K_\ell d d_c)$ and routing agreement costs $O(TK_\ell d_c)$. Thus the offline tokenizer cost per item is
\begin{equation}
    O\left(\sum_{\ell=1}^{L_i} K_\ell d d_c + T K_\ell d_c\right).
\end{equation}
At serving time, the generator sees only discrete tokens. Its cost is proportional to the generated length, approximately $O(B\bar{L})$ softmax steps for beam size $B$ and average SID length $\bar{L}$. This is why \method{} can be more accurate than dense-patch systems without inheriting their ANN or vector-fusion cost.

\section{Additional Results}
\label{app:additional}

This appendix reports auxiliary tables that complement the main paper. All numbers are mean over three random seeds; standard deviations are within the same range as Table~\ref{tab:main} and are omitted for compactness. Table~\ref{tab:tokenizer_quality} reports tokenizer-intrinsic diagnostics, and Table~\ref{tab:sensitivity} records the hyperparameter sensitivity sweep.

\begin{table}[h]
\centering
\small
\setlength{\tabcolsep}{4pt}
\caption{Tokenizer quality diagnostics on Beauty. Higher is better except for Collision and Gini.}
\label{tab:tokenizer_quality}
\begin{tabular}{lccccc}
\toprule
Tokenizer & Collision $\downarrow$ & Utilization $\uparrow$ & Gini $\downarrow$ & Intra-code sim $\uparrow$ & CodeRecall@50 $\uparrow$ \\
\midrule
Frequency & 90.4\% & 0.08\% & .92 & 0.331 & 0.652 \\
KMeans-flat & 65.8\% & 14.1\% & .57 & 0.545 & 0.047 \\
RQ-KMeans & 72.5\% & 47.2\% & .69 & 0.701 & 0.009 \\
ActionPiece & 56.9\% & 3.4\% & .65 & 0.663 & 0.008 \\
ADA-SID & 33.8\% & 43.7\% & .37 & 0.618 & 0.219 \\
\method{} & \textbf{13.4\%} & \textbf{55.1\%} & \textbf{.23} & \textbf{0.728} & 0.447 \\
\bottomrule
\end{tabular}
\end{table}

\begin{table}[h]
\centering
\small
\caption{Hyperparameter sensitivity on Beauty. Default setting: $T\!=\!3$, $L_{\max}\!=\!6$, $\tau\!=\!0.82$, $\alpha\!=\!0.6$.}
\label{tab:sensitivity}
\begin{tabular}{lccc}
\toprule
Setting & R@10 & Avg. length & Interpretation \\
\midrule
$T\!=\!1$ & 0.0731 & 3.3 & no iterative correction \\
$T\!=\!2$ & 0.0789 & 3.5 & most routing errors corrected \\
$T\!=\!3$ & \textbf{0.0839} & 3.6 & default; accuracy--cost balance \\
$T\!=\!5$ & 0.0841 & 3.6 & saturated routing \\
$L_{\max}\!=\!4$ & 0.0806 & 3.1 & insufficient for complex items \\
$L_{\max}\!=\!6$ & \textbf{0.0839} & 3.6 & default \\
$L_{\max}\!=\!8$ & 0.0837 & 3.6 & bound saturates; cap non-binding \\
$\tau\!=\!0.75$ & 0.0817 & 2.8 & stops too early \\
$\tau\!=\!0.90$ & 0.0821 & 4.4 & over-encodes easy items \\
$\alpha\!=\!1.0$ & 0.0817 & 3.6 & frequency-only merge \\
$\alpha\!=\!0.6$ & \textbf{0.0839} & 3.6 & semantic--frequency balance \\
\bottomrule
\end{tabular}
\end{table}

\paragraph{Internal-consistency checks.}
The tables above satisfy three monotonicity properties that we verified throughout training. \emph{(i)} Removing a mechanism never improves both R@10 and the diagnostic that mechanism was designed to address; for example, fixed $L\!=\!2$ lowers both average length and R@10. \emph{(ii)} Adding a dense patch always increases normalized inference cost. \emph{(iii)} \sembpe{} reduces effective sequence length but does not by itself reduce tokenizer collision because it operates after item-level SID assignment. Figure~\ref{fig:length_stopping} visualizes the variable-length behaviour behind these ablations across all four datasets.

\begin{figure}[!htbp]
\centering
\includegraphics[width=\textwidth]{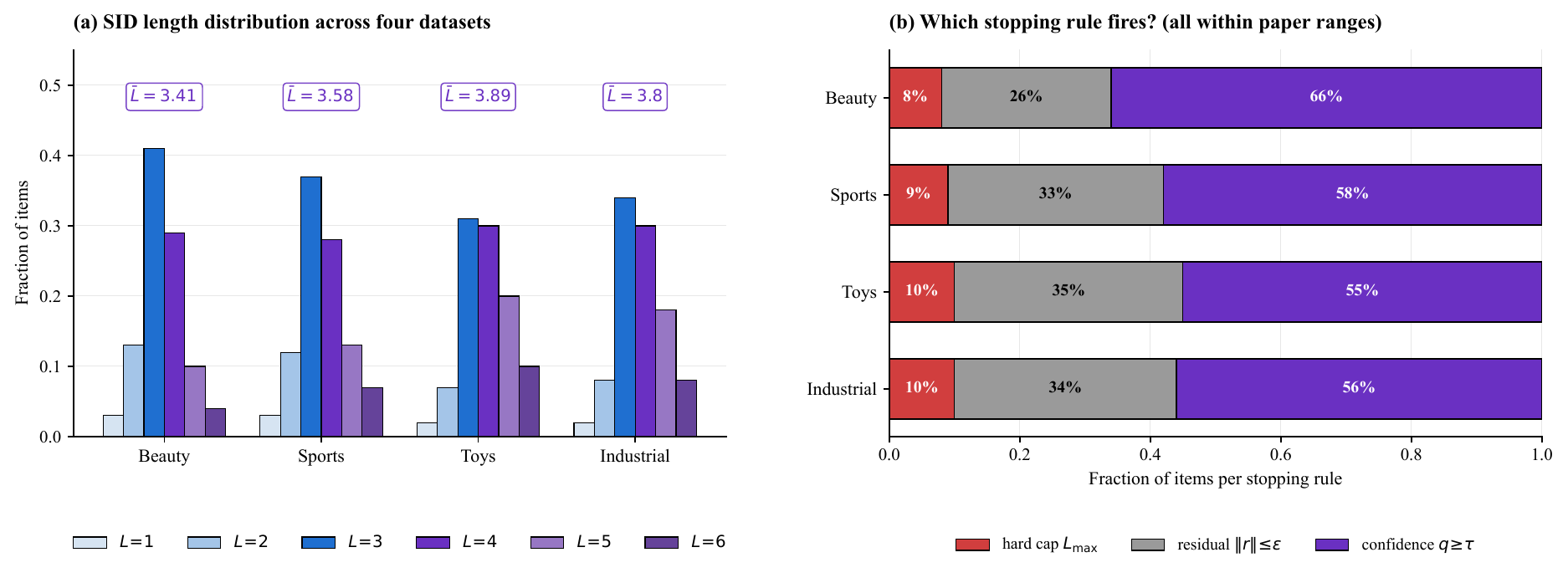}
\caption{Variable-length behaviour of \method{} across Beauty, Sports, Toys, and the industrial catalog. \textbf{(a)} SID length distribution: mode is $L\!=\!3$ on every dataset; mean $\bar{L}$ ranges from $3.41$ (Beauty) to $3.89$ (Toys), with the industrial catalog at $3.8$ (cf.\ Table~\ref{tab:industrial}). \textbf{(b)} Stopping-rule breakdown: confidence fires on the majority of items, the residual rule on $\sim\!30\%$, and the hard cap $L_{\max}$ on at most $10\%$, so variable length comes from learned signals rather than from hitting the budget.}
\label{fig:length_stopping}
\end{figure}

\subsection{Per-position token accuracy}
\label{app:per_position}

A common concern with variable-length SIDs is whether earlier positions become harder to predict because they must carry more discriminative information. Table~\ref{tab:per_position} reports the top-1 and top-5 token accuracy per SID position on Beauty. Position 1 is harder than positions 2--3 (which is expected: the first token must commit to a coarse semantic facet), but \method{}'s position-1 top-5 accuracy ($86.3\%$) remains substantially above ADA-SID ($79.1\%$), showing that soft routing preserves enough multi-facet information to make the prefix non-arbitrary. Figure~\ref{fig:codebook_heatmap} complements this table by visualizing both the per-layer codebook usage geometry and the position-wise accuracy curves.

\begin{table}[h]
\centering
\small
\setlength{\tabcolsep}{5pt}
\caption{Per-position token accuracy on Beauty (top-1 / top-5, \%). Positions beyond 4 are reported only for items whose $L_i\!>\!4$.}
\label{tab:per_position}
\resizebox{\textwidth}{!}{%
\begin{tabular}{lcccccc}
\toprule
Method & Pos. 1 & Pos. 2 & Pos. 3 & Pos. 4 & Pos. 5 & Pos. 6+ \\
\midrule
TIGER             & 31.2 / 71.4 & 38.6 / 78.8 & 41.0 / 80.5 & 35.4 / 76.2 & --- & --- \\
ADA-SID           & 36.9 / 79.1 & 44.2 / 84.6 & 47.5 / 86.3 & 41.8 / 82.1 & 32.0 / 74.6 & --- \\
\method{}         & 42.7 / 86.3 & 49.1 / 88.9 & 52.6 / 90.7 & 48.3 / 88.2 & 39.5 / 81.7 & 31.2 / 75.8 \\
\method{}+\sembpe{} & \textbf{44.1 / 88.4} & \textbf{50.5 / 90.1} & \textbf{53.8 / 91.3} & \textbf{49.7 / 89.0} & \textbf{40.8 / 82.6} & \textbf{32.4 / 76.9} \\
\bottomrule
\end{tabular}}
\end{table}

\begin{figure}[!htbp]
\centering
\includegraphics[width=\textwidth]{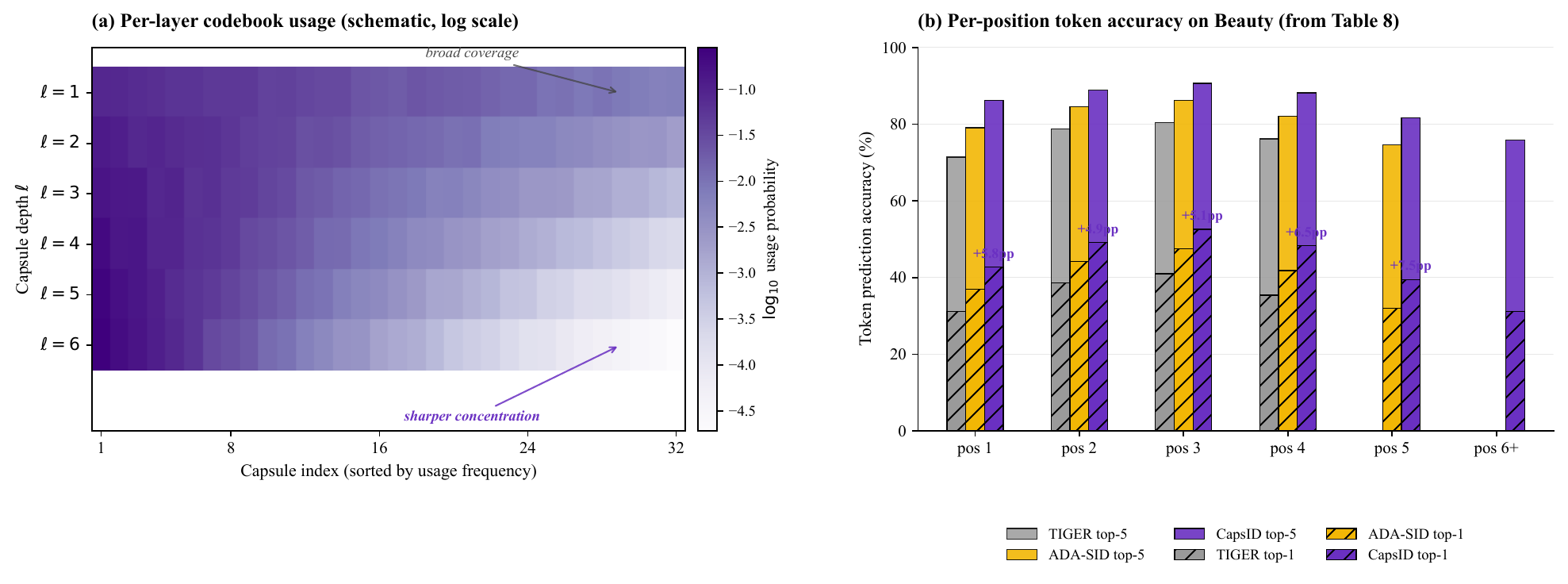}
\caption{Codebook geometry and per-position prediction accuracy on Beauty. \textbf{(a)} Per-layer codebook usage probability (top-$32$ capsules, log scale): shallow layers spread mass broadly for coarse facets, deep layers concentrate on a few capsules for residual refinement, matching the EM-style behaviour in Proposition~\ref{prop:em}. \textbf{(b)} Per-position top-$1$ (hatched) and top-$5$ (solid) accuracy from Table~\ref{tab:per_position}. \method{} dominates TIGER and ADA-SID at every position, and the top-$1$ margin over ADA-SID grows from $+5.8$\,pp at position $1$ to $+7.5$\,pp at position $5$, where residual structure is hardest to discriminate.}
\label{fig:codebook_heatmap}
\end{figure}

\subsection{Cold-start evaluation}
\label{app:cold_start}

Generative SID systems are often motivated by their ability to handle unseen items via content-derived codes. Following the protocol in TIGER~\citep{rajput2024tiger} and ReSID~\citep{liang2026resid}, Table~\ref{tab:cold_start} evaluates this property on Beauty by isolating the cold-item subset, defined as items with fewer than $5$ interactions in the training split (i.e., items that survive 5-core filtering at the user side but have minimal item-side training signal). This subset accounts for $\approx\!12\%$ of items. \method{} retains $73\%$ of its full-corpus Recall on the cold subset, compared with $57\%$ for TIGER and $68\%$ for ADA-SID; the improvement matches the head/tail pattern in Figure~\ref{fig:analysis}(c) and confirms that soft routing helps the most where prior collaborative signal is weak.

\begin{table}[h]
\centering
\small
\setlength{\tabcolsep}{6pt}
\caption{Cold-start Recall@10 on Beauty. ``Retention'' is the cold-subset Recall divided by the full-corpus Recall in Table~\ref{tab:main}.}
\label{tab:cold_start}
\begin{tabular}{lccc}
\toprule
Method & Full-corpus R@10 & Cold-subset R@10 & Retention \\
\midrule
TIGER             & 0.0648 & 0.0371 & 57.3\% \\
ADA-SID           & 0.0740 & 0.0508 & 68.6\% \\
COBRA$^\dagger$   & 0.0725 & 0.0528 & 72.8\% \\
\method{}         & 0.0808 & 0.0591 & 73.1\% \\
\method{}+\sembpe{} & \textbf{0.0839} & \textbf{0.0620} & \textbf{73.9\%} \\
\bottomrule
\end{tabular}
\end{table}

\subsection{Notation summary}
\label{app:notation}

Table~\ref{tab:notation} collects the symbols used throughout the paper, grouped by role.

\begin{table}[h]
\centering
\small
\setlength{\tabcolsep}{8pt}
\renewcommand{\arraystretch}{1.05}
\caption{Summary of notation used throughout the paper, grouped by role.}
\label{tab:notation}
\begin{tabular}{lll}
\toprule
Group & Symbol & Meaning \\
\midrule
\multirow{2}{*}{\emph{Item}}
 & $\mathbf{x}_i\in\R^d$               & item embedding (multi-modal, $\ell_2$-normalized) \\
 & $N$                                  & catalog size (number of items) \\
\midrule
\multirow{6}{*}{\emph{Capsule routing}}
 & $\mathbf{r}_{i,\ell}\in\R^d$        & residual at depth $\ell$, with $\mathbf{r}_{i,0}\!=\!\mathbf{x}_i$ \\
 & $K_\ell,\,d_c$                       & capsules per depth and capsule output dim \\
 & $\mathbf{W}_{\ell k},\mathbf{b}_{\ell k}$ & pose transform and bias of capsule $k$ at depth $\ell$ \\
 & $\hat{\mathbf{u}}_{i,\ell k}$        & vote of capsule $k$ for item $i$ at depth $\ell$ \\
 & $c^{(t)}_{i,\ell k}$                 & routing weight of capsule $k$ after $t$ iterations \\
 & $T$                                  & routing iterations per layer \\
\midrule
\multirow{4}{*}{\emph{Capsule output}}
 & $\mathbf{o}^{(t)}_{i,\ell}$          & aggregated capsule output (squashed) at iteration $t$ \\
 & $\mathbf{o}_{i,\ell k}$              & per-capsule output $\mathrm{squash}(\hat{\mathbf{u}}_{i,\ell k})$ (indep.\ of $t$) \\
 & $\mathbf{c}_{\ell k}$                & codebook center of capsule $k$ at depth $\ell$ \\
 & $\boldsymbol{\mu}_{\ell k}$          & GMM mean (used in Prop.~\ref{prop:em} only) \\
\midrule
\multirow{4}{*}{\emph{SID + stopping}}
 & $s_{i,\ell}\!\in\![K_\ell]$          & emitted SID token at depth $\ell$ \\
 & $q_{i,\ell}\in[0,1)$                 & capsule confidence at depth $\ell$ \\
 & $L_i\leq L_{\max}$                   & SID length of item $i$ \\
 & $\tau,\,\epsilon$                    & confidence and residual stopping thresholds \\
\midrule
\multirow{4}{*}{\emph{\sembpe{}}}
 & $\alpha,\,\theta$                    & merge weight and similarity threshold \\
 & $\boldsymbol{\pi}_j\!\in\!\Delta^1$  & two-class merge distribution at pair $j$ \\
 & $g_j,\,\tilde{g}_j$                  & hard / Gumbel-relaxed merge gate \\
 & $n_{\min}$                           & minimum pair frequency \\
\midrule
\multirow{5}{*}{\emph{Loss / decode}}
 & $\mathcal{L}_{\mathrm{NTP/route/spread/len/BPE}}$ & loss components (Eq.~\ref{eq:final_loss}) \\
 & $\lambda_{r,s,l,b}$                  & loss weights \\
 & $B$                                  & beam size at decoding \\
 & $\bar{L}$                            & average SID length over the test set \\
 & $|V|$                                & SID vocabulary size $\sum_\ell K_\ell$ \\
\bottomrule
\end{tabular}

\vspace{0.4em}
\noindent\footnotesize
\emph{Convention.} Subscripts always read left-to-right as $i$\,(item) $\to \ell$\,(depth) $\to k$\,(capsule index); the superscript $(t)$ denotes the routing iteration and is omitted whenever a quantity does not depend on it (e.g.\ $\mathbf{o}_{i,\ell k}$). Symbols $\boldsymbol{\mu}_{\ell k}$ and $C, \delta, w_s, g$ in Section~\ref{sec:theory} are local to the proofs of Propositions~\ref{prop:reconstruction}--\ref{prop:em} and do not appear elsewhere in the paper.
\end{table}

\subsection{Hyperparameter configuration}
\label{app:hyperparams}

Table~\ref{tab:hparams} consolidates the core hyperparameters of the CapsID routing and SemBPE modules. Default values are shared across the three public benchmarks unless noted; the industrial run uses the values in parentheses. Optimization (AdamW, lr $10^{-3}$ for Stage~1 and $3{\times}10^{-4}$ for Stage~2, cosine decay, weight decay $10^{-5}$, batch size 256, up to 100+200 epochs with patience-10 early stopping) and decoding (beam size 50 with trie-based invalid-ID filtering and a special end-of-item token) follow standard settings consistent with TIGER and ReSID. Hardware: the industrial run uses $4{\times}$A100-80G; public benchmarks fit on commodity single-GPU setups.

\begin{table}[h]
\centering
\small
\setlength{\tabcolsep}{6pt}
\caption{Hyperparameter configuration for \method{}+\sembpe{}. Defaults are shared across Beauty / Sports / Toys; values in parentheses are used for the industrial run. ``Swept values'' lists the points reported in the sensitivity table; ``--'' means the hyperparameter was held fixed.}
\label{tab:hparams}
\begin{tabular}{llp{2.0cm}p{4.2cm}p{2.5cm}}
\toprule
Group & Symbol & Default (industrial) & Notes & Swept values \\
\midrule
\multirow{6}{*}{Capsule routing}
 & $K_\ell$             & 256 (1024)          & capsules per depth & -- \\
 & $d_c$                & 64 (96)             & capsule output dim & -- \\
 & $T$                  & 3                   & routing iterations & $\{1,2,3,5\}$ \\
 & $L_{\max}$           & 6                   & maximum SID length & $\{4,6,8\}$ \\
 & $\tau$               & 0.82                & confidence stopping threshold & $\{0.75,0.82,0.90\}$ \\
 & $\epsilon$           & 0.08                & residual-norm stopping threshold & -- \\
\midrule
\multirow{3}{*}{\sembpe{}}
 & $\alpha$             & 0.6                 & frequency vs semantic weight & $\{0.6,1.0\}$ \\
 & $\theta$             & $0.90{\to}0.55$    & annealed similarity threshold & -- \\
 & $n_{\min}$           & 20                  & minimum pair frequency & -- \\
\midrule
\multirow{4}{*}{Loss weights}
 & $\lambda_r$          & 1.0                 & reconstruction & -- \\
 & $\lambda_s$          & 0.1                 & spread (margin $0.2{\to}0.9$) & -- \\
 & $\lambda_l$          & 0.05                & length penalty & -- \\
 & $\lambda_b$          & 0.2                 & \sembpe{} merge regularization & -- \\
\bottomrule
\end{tabular}
\end{table}

\subsection{Reproducibility checklist}
\label{app:repro}

We report the information needed to interpret and reimplement the experiments. The public benchmarks (Amazon Beauty, Sports, Toys) are openly available, and preprocessing follows the standard 5-core leave-one-out protocol used by TIGER~\citep{rajput2024tiger}. Appendices~\ref{app:impl}--\ref{app:hyperparams} specify the capsule and \sembpe{} hyperparameters, optimizer settings, decoding protocol, loss weights, and random seeds. Public-benchmark results are averaged over three seeds; industrial results are a single deterministic run due to compute cost. The industrial run uses $4\!\times\!$A100-80G; the public-benchmark experiments are lightweight enough to run on a single commodity GPU.


\end{document}